\documentclass[aps,prl,twocolumn,groupedaddress]{revtex4}
\usepackage{algorithm2e}
\usepackage{amsmath,amssymb,amsfonts,amsthm}
\usepackage{multirow}
\usepackage{verbatim}
\usepackage{url}
\usepackage{comment}
\usepackage{mathtools}
\usepackage{epsf,pgf,graphicx}
\usepackage{color}
\usepackage{tikz,pgf}
\usepackage{epstopdf}

\par

\begin{document}
\newcommand{\ket}[1]{\ensuremath{\left|#1\right\rangle}}
\newcommand{\bra}[1]{\ensuremath{\left\langle#1\right|}}
\newcommand\floor[1]{\lfloor#1\rfloor}
\newcommand\ceil[1]{\lceil#1\rceil}
\newtheorem{definition}{Definition}
\newtheorem{theorem}{Theorem}
\newtheorem{claim}{Claim}
\newtheorem{lemma}{Lemma}
\title{Device Independent Quantum Secret Sharing in Arbitrary Even Dimension}
\author{Sarbani Roy}
\email{sarbani.roy@iitkgp.ac.in}
\author{Sourav Mukhopadhyay}
\email{msourav@gmail.com}
\affiliation{Department of Mathematics, Indian Institute of Technology Kharagpur, West Bengal 721302, India.}

\begin{abstract}
We present a device independent quantum secret sharing scheme in arbitrary even dimension. We propose a $d$-dimensional $N$-partite linear game, utilizing a generic multipartite higher dimensional Bell inequality, a generalization of Mermin's inequality in the higher dimension. Probability to win this linear game defines the device independence test of the proposed scheme. The security is proved under causal independence of measurement devices and it is based on the polygamy property of entanglement. By defining $\epsilon_{cor}$-correctness and $\epsilon^c$-completeness for a quantum secret sharing scheme, we have also shown that the proposed scheme is $\epsilon_{cor}$-correct and $\epsilon^c$-complete.
\end{abstract}
\maketitle

\section{Introduction}
Quantum physics pledges the security of cryptographic protocols from an adversary having unbounded computational power. The journey of quantum cryptography was started in 1984, when Bennett and Brassard~\cite{BB84} proposed a Quantum Key distribution (QKD) protocol to share a common secret key between two distant users in a quantum mechanical way. Although it has been proven that the BB84 QKD protocol~\cite{BB84} is unconditionally secure~\cite{mayers,shor}, but the practical implementation of QKD confronted side channel attacks in which classical information gets leaked~\cite{bbbss,blms,gllskm}. To avert these type of attacks Mayers and Yao~\cite{mayersyao} proposed the idea of device independence in which the security of a scheme is based on a statistical test performed on the inputs and outputs of two spatially separated measurement devices. Since then numerous proposals~\cite{diqkd1,diqkd2,diqkd3,diqkd4,diqkd5,diqkd6,diqkd7,diqkd8,diqkd9,diqkd10,diqkd11} have been suggested for device independent QKD (DI-QKD). Although, recently researchers were became interested to explore device independent security outside QKD~\cite{outqkd1,outqkd2,outqkd3,outqkd4,outqkd5,outqkd6,outqkd7}, but it is very little known till today. A statistical test based on the quantum violation of Bell inequality~\cite{bell} is performed for the security of DI-QKD. Nowadays, quantum cryptography is observant of more multipartite higher dimensional schemes. Such schemes can be proven secure under device independence paradigm by employing the higher dimensional multi-particle Bell inequality.

Quantum nonlocality is evident by measuring two entangled particles in two spatially separated regions. Bell~\cite{bell} showed that the quantum correlations obtained by measuring two spatially separated entangled particles, can not be reproduced by local realism. Allowing more flexibility in local measurements than original Bell inequality; Clauser, Horne, Shimony, and Holt~\cite{chsh} proposed an inequality (CHSH inequality) which is maximally violated when each particle of a two-dimensional bipartite maximally entangled state is measured by two mutually unbiased bases. Then it was a natural question to generalize Bell inequality for a multipartite and/or multidimensional system. Mermin~\cite{mermin} derived such a violation for arbitrarily many qubits and showed that the degree of violation increases exponentially with the number of involved particles. Ardehali~\cite{arde}, Roy and Singh~\cite{rs}, Belinskii and Klyshko~\cite{bk} pursued in this direction further. Generalization of Bell inequality for the bipartite multidimensional system has been investigated in~\cite{bellmd1,bellmd2,bellmd3,bellmd4,bellmd5}. The long-awaited proof of the violation of Bell type correlations in a higher dimensional multipartite system was presented by Son {\it et al.}~\cite{son}. However, that was further generalized for more than two possible settings in the measurement device~\cite{ms1,ms2}. The applications of multipartite arbitrary dimensional Bell inequality in quantum information science is still to explore.

After QKD, quantum secret sharing (QSS) is the most remarkable example of quantum advantages. QSS is an adaption of classical secret sharing scheme~\cite{shamir} in the quantum world. In a secret sharing scheme, there is a dealer who wants to share a secret among a number of participants in such a way that only a set of authorized participants can recover the secret. Hillery {\it et al.}~\cite{Hillery} first proposed a QSS scheme using three-particle and four particle GHZ states of qubits. After that a large number of schemes on QSS have been proposed such as circular QSSs~\cite{zhou,lin,zhu}, dynamic QSSs~\cite{hsu,wang}, single particle QSSs~\cite{tav,kar}, graph state QSSs~\cite{markham,keet,sar}, verifiable QSSs~\cite{yang1,yang2,sl} and QSSs based on error correcting codes~\cite{cleve,sk}, phase shift operations~\cite{qin,du,liu} and quantum search algorithms~\cite{hsu1}. Gogioso~\cite{gig} proposed device independent quantum secret sharing scheme using Mermin-type contextuality. Although, it was initially claimed that the scheme is provably secure against non-signaling attackers but later it confronted some errors in the proof of main result~\cite{gig}. Our proposal is independent and very different from that proposal.

In this current draft, we have presented a device independent quantum secret sharing scheme in arbitrary even dimension. To design device independence test of the proposed scheme we have introduced an $N$-partite $d$-dimensional linear game based on the Bell violation~\cite{son}. It has been shown that to share and recover the secret securely, it is sufficient to pass the testing phase of the proposed scheme which is directly related to the winning probability of the proposed game. Security depends on the polygamy property of entanglement, which states that an $N$-partite GHZ state cannot be shared among more than $N$ parties~\cite{cbc}. By defining the notions of $\epsilon_{cor}$-correct and $\epsilon^{c}$-complete QSS scheme, we have derived the correctness and completeness of the proposed scheme. Completeness and security is proved under causal independence assumption on measurement devices. As per our best knowledge, this is the first quantum cryptographic scheme that uses higher dimensional multipartite Bell inequality to prove security in the device independent paradigm.

\section{$N$-partite $d$-dimensional XOR game}
{\it $N$-partite $d$-dimensional Bell Inequality:} In this section, we have proposed a linear game depending on the multipartite arbitrary dimensional Bell inequality proposed by Son {\it et al.}~\cite{son} with a slight modification. The detailed discussion can be found in Appendix A. Let us consider the scenario that, each of $N$ observers independently
chooses one of two variables (denoted by $P_{0,j}$ and $P_{1,j}$ for the
$j^{th}$ observer), each of which takes a value from a set generated by the $d$th root of unity $\omega$ over complex field. The corresponding Bell function is given by:
\begin{equation}
{\mathcal{B}}=\frac{1}{2^N}\sum_{n=1}^{d-1}\left<\prod_{j=1}^{N}(P_{0,j}^n+\omega^{-\frac{n}{2}}P_{1,j}^n)\right>
+ \text{C. C.},
\end{equation}
where C. C. stands for complex conjugate.

From the classical viewpoint, the symbol $\left<\cdot\right>$
denotes the statistical average over many runs. The theory of local
realism implies that,
\begin{equation*}
{\mathcal{B}} \leq
     \begin{cases}
       d(2^{-\frac{N}{2}}+2^{-1})-1, & \text{if}\ N~\text{and}~d~\text{both are even} \\
       d(2^{-\frac{N+1}{2}}+2^{-1})-1, & \text{if}\ N~\text{is odd and}~d~\text{is even} \\
       d-1, & \text{otherwise}
     \end{cases}
\end{equation*}
In case of quantum mechanical description, the statistical average
is replaced by quantum average on a given state $\ket{\psi}$. The
quantum expectation of $\mathcal{B}$
takes the maximum value $(d-1)$. One can clearly see that the constraint on $\mathcal{B}$
imposed by local realism is violated by quantum mechanics for an even dimensional system. In
particular, the maximum quantum violation is reached by a maximally
entangled state
$\ket{\psi}=\frac{1}{\sqrt{d}}\sum_{\alpha=0}^{d-1}\ket{\alpha}_1\ket{\alpha}_2\hdots\ket{\alpha}_N,$
and for the observable operators $P_{x_k,k}=\sum_{\alpha=0}^{d-1}\omega^{\alpha}\ket{\alpha}_{x_k
x_k}\bra{\alpha},$
where $x_k\in\{0,1\},$ $\ket{\alpha}_{0}=\sum_{\beta=0}^{d-1}\omega^{-\alpha\beta}\ket{\beta}$,
$\ket{\alpha}_{1}=\sum_{\beta=0}^{d-1}\omega^{-(\alpha-\frac{1}{2})\beta}\ket{\beta}$ and $\ket{\alpha}$ is the eigen vector of generalized $Z$ operator
corresponding to the eigen value $\omega^{\alpha}$. In the case of two dimensions, the observable operators
$P_{0,k}$ and $P_{1,k}$ will be reduced to Pauli $X$ and $Y$
operators. Thus one can call the
observable operators $P_{0,k}$ and $P_{1,k}$ as generalized $X$ and
$Y$ operator. Now, we can rewrite $\mathcal{B}$ as:
\begin{eqnarray*}
\mathcal{B}
& = &\frac{1}{2^{N-1}}\sum_{n=1}^{d-1}\sum_{x\in\{0,1\}^N}\omega^{-nf(x)}\left<\prod_{j=1}^{N}P_{x_j,j}^n\right>,
\end{eqnarray*}
where $x_i\in\{0,1\}$ is the $i^{th}$ bit of $x$ and $\omega^{\bot}=0$
$f:\{0,1\}^N\rightarrow\{0,1,\hdots,d-1,\bot\}$ is a function such
that it takes the value $\bot$ for a $x\in\{0,1\}^N$, if the binary representation of $x$ contains odd no of 1's and $f(x)=j$, if the no of 1's in the binary representation of x is $2j(\text{mod }2d)$.

{\it Proposal of the game:} Let us consider an XOR game between $N$ spatially separated players $Bob_1$, $Bob_2$, $\hdots$, $Bob_N$ who are not allowed
to communicate during the game. The players can communicate before
the game starts and discuss their strategy (or send physical systems
to each other). Now, for each $k\in[N]$ (where $[N]=\{1,\hdots, N\})$, a question $x_k\in\{0,1\}$
is asked to a player $Bob_k$ uniformly at random and independent to
the questions $x_1$, $\hdots$, $x_{k-1}$, $x_{k+1}$, $\hdots$, $x_N$
asked to the other players. $Bob_k$ gives an answer
$a_k\in\{0,1,\hdots,d-1\}$ to a question $x_k$. The players will win
the game if
$$\bigoplus_{i=1}^{N}a_i = f(x_1, \hdots x_N),$$ where the function
$f$ is the same function defined as above and $\bigoplus$ is used to define the addition modulo $d$. It is obvious that, they will never win the game when $f(x_1, \hdots, x_n)
= \bot$.

{\it Quantum strategy to win the game:} Before the game starts, players $Bob_1$, $Bob_2$, $\hdots$, $Bob_N$ share an $N$-partite
$d$-dimensional GHZ state
$\ket{\psi}=\sum_{\alpha=0}^{d-1}\ket{\alpha}_1\ket{\alpha}_2\hdots\ket{\alpha}_N$. Here, the suffixes clarify that $Bob_k$ ($k\in[N]$)
holds the $k^{th}$ particle of $\ket{\psi}$. For a question $x_k\in\{0,1\}$, $Bob_k$
measures his particle with the observable $P_{x_k,k}$ and gets the
eigen value $\omega^{a_k}$ as a measurement result. Then $Bob_k$
outputs $a_k$ as an answer to the question $x_k$.

{\it Probability to win the game:} We will calculate the probability
to win the game $p_w$ using the approach of Murta {\it et al.}~\cite{murta}. As the
question $x_k$ is chosen from $\{0,1\}$ uniformly at random, then
$\Pr(x_1, \hdots, x_N) = 2^{-N}.$ Now,
{\scriptsize
\begin{eqnarray*}
p_{w} & = & \sum_{x\in\{0,1\}^N} \Pr(x_1, \hdots,
x_N)\Pr\left(\bigoplus_{i=1}^{N}a_i=f(x_1, \hdots, x_N)|x_1, \hdots, x_N\right)\\
& = & 2^{-N}\sum_{j=o}^{d-1}\sum_{\substack{x \in \{0,1\}^N\\
\text{s.t. }f(x)=j}} \Pr\left(\bigoplus_{i=1}^{N}a_i=j|x_1, \hdots, x_N\right)\\
 & & + 2^{-N}\underbrace{\sum_{\substack{x \in \{0,1\}^N\\
\text{s.t. }f(x)=\bot}} \Pr\left(\bigoplus_{i=1}^{N}a_i=\bot|x_1, \hdots, x_N\right)}_{=0}
\end{eqnarray*}}

Let us consider the Abelian group $\mathbb{Z}_d$. Now,
we can apply the Fourier transform on $\mathbb{Z}_d$ and rewrite
the success probability $p_w$ as
{\scriptsize
\begin{equation}
p_{w} =\frac{1}{d}\left(1+2^{-N}\sum_{j=0}^{d-1}\sum_{\substack{x \in \{0,1\}^N\\
\text{s.t. }f(x)=j}}\sum_{n=1}^{d-1}\chi_n(j)\left<A_{x_1,1}^n\hdots
A_{x_N,N}^n\right>\right),
\end{equation}}
where
{\scriptsize
\begin{eqnarray*}
& &\left<A_{x_1,1}^n\hdots A_{x_N,N}^n\right>\\
& & =\sum_{a_1, \hdots, a_N \in \mathbb{Z}_d}\overline{\chi_n}(a_1)\hdots
\overline{\chi_n}(a_N)\Pr(a_1,\hdots,a_N|x_1,\hdots,x_N)
\end{eqnarray*}}
and $\chi_n$'s are the characters of the Abelian group
$\mathbb{Z}_d$ associated with the game. In particular,
$\chi_n(j)=\omega^{-nj}$. Quantum strategy to win the game is defined by the
set of projective measurements $\{P_{x_1,1}^{a_1}\}$, $\hdots$,
$\{P_{x_N,N}^{a_N}\}$ being performed on a $N$-partite entangled state
$\ket{\psi}$. In this case, the generalized correlators are defined
by $\left<A_{x_1,1}^n\hdots
A_{x_N,N}^n\right>=\bra{\psi}A_{x_1,1}^n\hdots
A_{x_N,N}^n\ket{\psi},$ where the operators $A_{x_k,k}^n$ are
defined as
\begin{equation*}
A_{x_k,k}^n =
\sum_{a_k=0}^{d-1}\overline{\chi_n}(a_k)P_{x_k,k}^{a_k}
= (P_{x_k,k})^n.
\end{equation*}
Here we have used $P_{x_k,k}=\sum_{a_k=0}^{d-1}
\omega^{a_k} P_{x_k,k}^{a_k}$. Thus,
\begin{equation}
\left<A_{x_1,1}^n\hdots A_{x_N,N}^n\right>=\left<P_{x_1,1}^n \hdots
P_{x_N,N}^n\right>.
\end{equation}
Finally, by substituting $(3)$ in $(2)$, we get

{\scriptsize
\begin{eqnarray*}
p_{w}
& = & \frac{1}{d}\left(1+2^{-N}\sum_{x \in
\{0,1\}^N}\sum_{j=0}^{d-1}\sum_{n=1}^{d-1}\omega^{-n
f(x)}\left<P_{x_1,1}^n\hdots
P_{x_N,N}^n\right>\right)\\
& = &\frac{1}{d}\left(1+\frac{1}{2}(d-1)\right).
\end{eqnarray*}}

The final equality follows from the fact that $\omega^{\bot}=0$ and the Bell function $\mathcal{B}$ attains its maximum quantum value when a maximally entangled state
$\ket{\psi}$ is measured by generalized $X$ or $Y$ operators uniformly at random. As for an even $d$, the constraint on classical expectation value of $\mathcal{B}$ is violated by quantum mechanics thus the probability to win the game by using above quantum strategy will be greater than any classical strategy.

One may note that for $d=N=2$; $p_{w}=\frac{3}{4}$, which does
not match the CHSH-value as the function $\mathcal{B}$ does not generalize CHSH function~\cite{chsh}, it actually generalizes the Mermin's function~\cite{mermin} in arbitrary dimension. Thus, the winning probability of the game introduced in this study fails
to generalize the same of CHSH game. It remains an
open question to generalize the exact CHSH inequality~\cite{chsh} for a
$d$-dimensional $N$-partite system with two possible measurement settings.

\medskip
\section{The Protocol}

In this section, we will use the game introduced in the previous
section to propose a device independent quantum secret sharing
scheme (DI-QSS). In this scheme, a dealer Alice wants to share a classical secret among $N-1$ participants $Bob_1$, $Bob_2$, $\hdots$, $Bob_{N-1}$ in such a way that only all of the Bobs together can recover the secret. In particular, Alice's secret $S\in\mathbb{Z}_d$, where $d$ is even. Before describing the proposed scheme, we first enumerate a minimal set assumptions determining the security of the scheme.\\
{\it Assumptions:}
\begin{enumerate}
\item Alice and Bobs' laboratories are perfectly isolated from outside (in particular from Eve) such that any unintended information cannot go outside the labs.
\item Each party holds a trusted random number generator.
\item Each of Alice and $Bob_1$, $Bob_2$, $\hdots$, $Bob_{N-1}$ has a measurement device in their laboratory with two inputs. Each input has $d$ outputs. The measurement devices are causally independent. Otherwise, the measurement devices are arbitrary and therefore could be prepared by an eavesdropper.
\end{enumerate}

The proposed scheme can be divided into two phases. Where the first phase guarantees the device independence, the second phase shares and reconstructs the secret.
\begin{widetext}
\begin{enumerate}
\item For every round $i\in[M]$:
    \begin{itemize}
    \item A dealer Alice and $N-1$ participants $Bob_1$, $\hdots$, $Bob_{N-1}$ shares an $N$-partite $d$-dimensional GHZ state $$\ket{\psi}_i=\left(\frac{1}{\sqrt{d}}\sum_{\alpha=0}^{d-1}\ket{\alpha}_1\hdots\ket{\alpha}_N\right)_{i},$$ where $d$ is an even positive integer.
    \item Alice picks a random bit $T_i$ with $\Pr(T_i=1)=\mu$ and publicly communicates the choice of $T_i$ with Bobs.
    \item For $T_i=1$, Alice randomly picks $x_i\in\{0, 1\}$ and for each $k\in[N]$, $Bob_k$ picks $y_{ki}\in\{0,1\}$ uniformly at random. They input $x_i$, $y_{1i}$, $\hdots$, $y_{(N-1)i}$ to the respective measurement devices and record the outputs as $a_i$, $b_{1i}$, $\hdots$, $b_{(N-1)i}$ respectively. Then Alice and Bobs announce their input and output pairs. They define a random variable $C_i$ by,
        \begin{equation*}
        C_i =
        \begin{cases}
        1, & \text{if}\ \text{they win the N-partite d dimensional XOR game,} \\
        0, & \text{otherwise}
        \end{cases}
        \end{equation*}
    \end{itemize}
\item Testing: Alice and Bobs calculate $C=\frac{1}{\sum_{i}T_i}\left(\sum_{i=0}^{M}C_i\right)$. They abort the scheme if $C<p_w-\eta$, where $\eta$ is the noise tolerance. Otherwise, they proceed to the next steps to share and reconstruct the secret.
\item Sharing the secret: For each $i\in[M]$ such that $T_i=0$, each of Alice and Bobs measures their particle with generalized $X$-operator i.e., inputs $0$ to their measurement devices and stores their output as $S_{A_i}$, $S_{B_{1i}}$, $\hdots$, $S_{B_{(N-1)i}}$ respectively. Suppose Alice's secret corresponding to the round $i$ is $S_i\in\mathbb{Z}_d$, then she calculates $\acute{S_i} = S_i \bigoplus S_{A_i}$ and announces $\acute{S_i}$.
\item Secret recovery: Bobs calculate $\hat{S_i} = (\acute{S_i}\bigoplus_{k=1}^{N-1}S_{B_{ki}})$ and gets the secret $S_i$.
\item Error correction: Alice picks a hash function $h(\cdot)$ from a family of two-universal hash functions uniformly at random. She computes a dit-string $Z$ of length $\epsilon_{EC}$ by applying $h(\cdot)$ on the string of secrets and sends $Z$ to Bobs together with the choice of the hash function $h(\cdot)$. Bobs verify the hash value calculated from their recovered secret with the hash value sent by Alice. If it matches, they consider the recovered secret as the dealer's secret. Otherwise, they discard the secret.
\end{enumerate}
\begin{center}{\bf Scheme 1.} Proposed Device Independent Quantum Secret Sharing Scheme (DIQSS)\end{center}
\end{widetext}

\section{Correctness, completeness and security analysis}

{\it Correctness:} After passing the testing phase of the proposed scheme, each of Alice and Bobs holds a particle of an
$N$-partite $d$-dimensional GHZ state $\ket{\psi}_i$ corresponding to the round $i\in[M]$ with
$T_i=0$. After measuring their own particle along with generalized
$X$-operator they get a measurement results as $S_{A_i}$ and $S_{B_{ki}}$ (for $k\in[N-1]$) with a nonzero probability if $\bigoplus_{k=1}^{N-1}S_{B_{ki}}=-S_{A_i}.$ This condition contributes to recover shared secret perfectly. So, we can state the following theorem:

\begin{theorem}
In the absence of effective eavesdropping the dealer's secret can be perfectly reconstructed
by the participants for an honest implementation of secret distribution
and recovery phase of the proposed scheme.
\end{theorem}

\begin{definition}
A quantum secret sharing scheme is called $'$correct$'$, if, for any strategy of the adversary, $\hat{S} = S$, where $S$ is the dealer's secret and $\hat{S}$ is the secret recovered by the participants. It is called $\epsilon_{cor}$-correct, if it is $\epsilon_{cor}$-indistinguishable from a correct protocol.
\end{definition}

\begin{theorem}
The proposed scheme is $\epsilon_{cor}$-correct for a choice of the length of hash value $\epsilon_{EC}$ such that $\epsilon_{EC}\geq\log_{d}\left(\frac{1}{\epsilon_{cor}}\right)$.
\end{theorem}

{\it Completeness:} A QSS scheme is called complete if the probability to abort the scheme is very small in case of an honest implementation of the scheme. By the term $'$honest implementation$'$, we mean that none of the parties deviate from the scheme and there is no eavesdropping.

\begin{definition}
A quantum secret sharing scheme is called $\epsilon^c$-complete if there exists an honest implementation of the scheme such that the probability of aborting the scheme is less than $\epsilon^c$.
\end{definition}

In the next result, we show that the condition of being the proposed scheme $\epsilon^c$-complete depends on the noise tolerance $\eta$, the total number of rounds $M$ and the fraction ($\mu$) of rounds chosen for the testing phase.
\begin{theorem}
The proposed scheme is $\epsilon^c$-complete for $\epsilon^c\geq(1-\mu(1-\exp(2\eta^2)))^M + \eta.$
\end{theorem}

{\it Security:} The security of the proposed QSS scheme is based on a property of quantum entanglement called polygamy which states that all maximally entangled states
are (classically and quantically) uncorrelated with any
other system~\cite{cbc}. First, one can use the correlations obtained in the testing phase of the proposed scheme to guarantee that the shared entangled state is of the form $\ket{\psi}=\sum_{\alpha=0}^{d-1}\ket{\alpha}_1\hdots\ket{\alpha}_N$. Then based on the polygamy property of entanglement, it concludes that these correlations must be independent of any information that the eavesdropper Eve can obtain. This ensures the security of all measurement devices (held by Alice and Bobs). As the statistical state is performed with a set of entangled particles and the remaining particles are used to distribute and recover the secret, then the security is highly dependent on the following theorem.
\begin{theorem}
For large $M$, Alice and Bobs can proceed for the share distribution and reconstruction phases securely, if the testing phase of the proposed scheme is successful.
\end{theorem}
The above discussions conclude that qualifying the testing phase is sufficient for the security of the proposed scheme. We refer Appendix B for the detailed proof of above theorems.

\section{Discussion and Conclusion}
We have presented a device independent quantum secret sharing scheme in the arbitrary even dimension. An $N$-partite $d$-dimensional XOR game has been proposed utilizing the violation of generic $N$-partite $d$-dimensional Bell inequality. One obvious generalization is to enhance the scheme to any (even or odd) dimension. As we have described earlier that the proposed game does not generalize the CHSH game. So, it remains an open question that, how to propose a multipartite CHSH game in arbitrary dimension by using a variant of multipartite higher dimensional Bell inequality with two possible measurement settings. Device independence test is based on the winning probability of the proposed linear game. We have shown that this testing phase is enough to ensure that, the entangled particles are in the desired form and measurement devices are also independent from any eavesdropper. The security of the scheme is based on the polygamy property of entanglement which states that an $N$-partite maximally entangled state cannot be distributed among more than $N$ parties. Introducing the definition of $\epsilon_{cor}$-correctness and $\epsilon^c$-completeness for quantum secret sharing scheme, we have shown that the proposed secret sharing scheme is $\epsilon_{cor}$-correct and $\epsilon^c$-complete. But, for the sake of security and completeness, we have assumed that the measurement devices are causally independent. This scheme can be improved by relaxing the causal independence assumption and propose a fully device independent quantum secret sharing scheme. Finally, we also remark that the proposed scheme can be easily adopted for the other variants of multipartite higher dimensional Bell inequality.

\section{ACKNOWLEDGEMENT}
One of the author (SR) acknowledges the support from the institute in the form of institute research fellowship (Grant no: IIT/Acad/PGS$\&$R/F.II/2/15/MA/90J03) of Indian Institute of Technology Kharagpur.

\section{Appendix A: Bell Inequalities for Multipartite Arbitrary Even Dimensional System}
Son {\it et al.}~\cite{son} have proposed generic Bell inequalities for
multipartite arbitrary dimensional system. As in the case of a
two-dimensional bipartite system, quantum mechanics
outperforms the constraint on the correlations between subsequent
measurements on the particles given by any local realistic theory
for a multipartite arbitrary even dimensional system. Let us
consider the scenario that, each of $N$ observers independently
chooses one of two variables (denoted by $A_j$ and $B_j$ for the
$j^{th}$ observer), each of which takes a value from a set generated by the $d$th root of unity $\omega$ over complex field.
 The generic Bell function
presented in~\cite{son} is:
\begin{equation}
\tilde{\mathcal{B}}=\frac{1}{2^N}\sum_{n=1}^{d-1}\left<\prod_{j=1}^{N}(A_j^n+\omega^{\frac{n}{2}}B_j^n)\right>
+ \text{C. C.},
\end{equation}
where C. C. stands for complex conjugate.

From the classical viewpoint, the symbol $\left<\cdot\right>$
denotes the statistical average over many runs. The theory of local
realism implies that,
\begin{equation*}
\tilde{\mathcal{B}} \leq
     \begin{cases}
       d(2^{-\frac{N}{2}}+2^{-1})-1, & \text{if}\ N~\text{and}~d~\text{both are even} \\
       d(2^{-\frac{N+1}{2}}+2^{-1})-1, & \text{if}\ N~\text{is odd and}~d~\text{is even} \\
       d-1, & \text{otherwise}
     \end{cases}
\end{equation*}
In case of quantum mechanical description, the statistical average
is replaced by quantum average on a given state $\ket{\psi}$. The
quantum expectation $\tilde{\mathcal{B}_q}$ of $\tilde{\mathcal{B}}$
takes the maximum value $(d-1)$ i.e. $\tilde{\mathcal{B}_q} \leq
(d-1)$. One can clearly see that the constraint on $\mathcal{B}$
imposed by local realism, formally known as generic Bell inequality,
is violated by quantum mechanics for an even dimensional system. In
particular, the maximum quantum violation is reached by a maximally
entangled state
$\ket{\psi}=\frac{1}{\sqrt{d}}\sum_{\alpha=0}^{d-1}\ket{\alpha}_1\ket{\alpha}_2\hdots\ket{\alpha}_N,$
and for the observable operators $\widehat{V}=\sum_{\alpha=0}^{d-1}
\omega^{\alpha}\ket{\alpha}_{VV}\bra{\alpha}$, where $V\in\{A,B\}$,
$\ket{\alpha}_A=\sum_{\beta=0}^{d-1}\omega^{-\alpha\beta}\ket{\beta}$,
$\ket{\alpha}_B=\sum_{\beta=0}^{d-1}\omega^{-(\alpha+\frac{1}{2})\beta}\ket{\beta}$
and $\ket{\alpha}$ is the eigen vector of generalized $Z$ operator
corresponding to the eigen value $\omega^{\alpha}$. Here
$\widehat{A}$ and $\widehat{B}$ are the quantum observable operators
analogous to the classical variable $A$ and $B$.

For our purpose, we will use a slight different form of the above Bell function (4), as:
\begin{equation}
\mathcal{B}=\frac{1}{2^N}\sum_{n=1}^{d-1}\left<\prod_{j=1}^{N}(A_j^n+\omega^{-\frac{n}{2}}B_j^n)\right>
+ \text{C. C.}.
\end{equation}

One can easily verify that the classical and quantum bounds for
$\mathcal{B}$ are same as of $\tilde{\mathcal{B}}$. Note that, to
reach the maximum violation of $\mathcal{B}$, it is
important to use an observable $\widehat{B}$, with an eigen vector
$\ket{\alpha}_B=\sum_{\beta=0}^{d-1}\omega^{-(\alpha-\frac{1}{2})\beta}\ket{\beta}$,
corresponding to the eigen value $\omega^{\alpha}$. By choosing
$\widehat{A}_j=P_{0,j}$ and $\widehat{B}_j=P_{1,j}$, one can rewrite
the Bell function (5) as:
\begin{eqnarray*}
\mathcal{B}
& = &\frac{1}{2^{N-1}}\sum_{n=1}^{d-1}\sum_{x\in\{0,1\}^N}\omega^{-nf(x)}\left<\prod_{j=1}^{N}P_{x_j,j}^n\right>,
\end{eqnarray*}
where $x_i\in\{0,1\}$ is the $i^{th}$ bit of $x$,
$f:\{0,1\}^N\rightarrow\{0,1,\hdots,d-1,\bot\}$ is a function such
that it takes the value $\bot$ for half of the elements in the domain
and $\omega^{\bot}=0$.
\section{Appendix B: Proof of the theorems}
\begin{theorem}
In the absence of effective eavesdropping the dealer's secret can be perfectly reconstructed
by the participants for an honest implementation of secret distribution
and recovery phase of the proposed scheme.
\end{theorem}

\begin{proof}
After passing the device independence test, for every $i\in[M]$ with
$T_i=0$, Alice and Bobs share an
$N$-partite $d$-dimensional GHZ state $\ket{\psi}_i$ and each of
them measures their own particle along with generalized
$X$-operator. For simplicity, we will write $\ket{\psi}$ in place of
$\ket{\psi}_i$. Now the operator of Alice and all Bobs together can
be expressed as

\begin{widetext}
{\scriptsize
\begin{eqnarray*}
X_A\otimes X_{B_1}\otimes\hdots\otimes X_{B_{N-1}} & =
&\left(\sum_{S_A=0}^{d-1}\omega^{S_A}\ket{S_A}_{0\text{
}0}\bra{S_A}\right)\otimes\left(\sum_{S_{B_1}=0}^{d-1}\omega^{S_{B_1}}\ket{S_{B_1}}_{0\text{
}0}\bra{S_{B_1}}\right)\\
&
 &\otimes\hdots\otimes\left(\sum_{S_{B_{N-1}}=0}^{d-1}\omega^{S_{B_{N-1}}}\ket{S_{B_{N-1}}}_{0\text{
}0}\bra{S_{B_{N-1}}}\right)\\
& = &
\sum_{S_A,S_{B_1},\hdots,S_{B_{N-1}}=0}^{d-1}\omega^{\left(S_A+S_{B_1}+\hdots+S_{B_{N-1}}\right)}
\ket{{S_A S_{B_1}\hdots S_{B_{N-1}}}}_{0\text{ }0} \bra{{S_A
S_{B_1}\hdots S_{B_{N-1}}}}.
\end{eqnarray*}}
\end{widetext}
And,  {\scriptsize
\begin{eqnarray*}
P & = & \left<\psi\ket{{S_A S_{B_1}\hdots S_{B_{N-1}}}}_{0\text{
}0} \bra{{S_A S_{B_1}\hdots S_{B_{N-1}}}}\psi\right>\\
& = & \left|{\left<\psi|S_A S_{B_1}\hdots
S_{B_{N-1}}\right>_0}\right|^2\\
& =
&\frac{1}{d^{N-1}}\left|\sum_{\alpha=0}^{d-1}\omega^{\left(S_A+S_{B_1}+\hdots
S_{B_{N-1}}\right)\alpha}\right|^2\\
& = &\begin{cases}
        \frac{1}{d^{N-1}}, & \text{if}\ (S_A+S_{B_1}+\hdots
S_{B_{N-1}})=0 \\
        0, & \text{otherwise}
        \end{cases}
\end{eqnarray*}}

Thus the condition to get a measurement result $S_A$ from the
observable $X_A$ and $S_{B_k}$ from $X_{B_k}$ (for $k\in[N-1]$) by
measuring the entangled state $\ket{\psi}$ with a nonzero
probability is $(S_A+S_{B_1}+\hdots +S_{B_{N-1}})=0(\text{mod}~d)$
i.e.,
$$\sum_{k=1}^{N-1}S_{B_k}=-S_A(\text{mod}~d).$$
For a secret $S$, Alice declares $\acute{S}=S+S_A(\text{mod}~d).$
After Alice's declaration Bobs calculate
\begin{eqnarray*}
\hat{S} & = & \acute{S}+\sum_{k=1}^{N-1}S_{B_k}(\text{mod}~d)\\
& = & \acute{S}-S_A(\text{mod}~d)\\
& = & S(\text{mod}~d).
\end{eqnarray*}
Thus, in the absence of effective eavesdropping the dealer's secret can be perfectly reconstructed
by the participants by an honest implementation of secret distribution
and recovery phase of the proposed scheme.
\end{proof}

\begin{theorem}
The proposed scheme is $\epsilon_{cor}$-correct for a choice of the length of hash value $\epsilon_{EC}$ such that $\epsilon_{EC}\geq\log_{d}\left(\frac{1}{\epsilon_{cor}}\right)$.
\end{theorem}

\begin{proof}
In the error correction phase of the scheme, Alice picks a hash function $h(\cdot)$ from a family of two-universal hash functions uniformly at random. Then, by the definition of two-universal hash function, two hash values $h(S)$ and $h(\hat{S})$ will coincide for two different values of $S$ and $\hat{S}$, with a probability atmost $d^{-\epsilon_{EC}}$, i.e.,$$\Pr(h(\hat{S})=h(S)|\hat{S}\neq S)\leq d^{-\epsilon_{EC}}.$$Thus the probability that the scheme does not abort at the error correction phase for $\hat{S}\neq S$ is atmost $\epsilon_{cor}$ for a choice of $\epsilon_{EC}$ such that $\epsilon_{EC}>\log_{d}(\frac{1}{\epsilon_{cor}})$. Hence the proposed scheme is $\epsilon_{cor}$-correct.
\end{proof}

\begin{theorem}
The proposed scheme is $\epsilon^c$-complete for $\epsilon^c>(1-\mu(1-\exp(2\eta^2)))^M + \eta.$
\end{theorem}
\begin{proof}
The proposed scheme can be aborted either in testing phase or in error correction phase. The probability to abort in the testing phase is given by
{\scriptsize
\begin{eqnarray*}
P_{test} & = & \Pr\left(\sum_{i}C_i<(p_w-\eta)\sum_i T_i\right)\\
& = & \sum_{j=0}^{M}\Pr\left(\sum_{i}C_i<(p_w-\eta)j|\sum_i T_i=j\right)\Pr\left(\sum_i T_i=j\right)
\end{eqnarray*}}
Devices are causally independent i.e., each use of device is independent of previous use, thus the random variables $C_i$ ($i\in[M]$) are independently and identically distributed. The expectation value of $C_i$ is $p_w$. Thus by using Hoeffding bound~\cite{Hoeffding}, we calculate $\Pr(\sum_{i}C_i<(p_w-\eta)j|\sum_i T_i=j)<\exp(-2\eta^2 j)$. As $T_i$ follows Bernoulli distribution with $\Pr(T_i=1)=\mu$, then $\Pr(\sum T_i=j)=\binom Mj \mu^{j}(1-\mu)^{M-j}$. Hence
{\scriptsize
\begin{eqnarray*}
P_{test} & = & \sum_{j=0}^{M}\binom M j \mu^{j}(1-\mu)^{M-j}\exp(-2\eta^2 j)\\
& = & (1-\mu(1-\exp(2\eta^2)))^M.
\end{eqnarray*}}
For the case of aborting in error correction phase, we have assumed that the scheme has passed the testing phase. Which implies that the particles are entangled in the specified form and measurement devices are also free from eavesdropping. In the previous theorem, we have shown that in the absence of any active eavesdropper the shared and the recovered secret is same except the interference of noise and from the property of the hash function, it is clear that if $\hat{S}=S$, then $h(\hat{S})=h(S).$ Thus in case of honest implementation of the scheme, the probability to abort in error correction phase ($P_{EC}$) is bounded above by the noise tolerance $\eta$ i.e., $P_{EC}\leq\eta$.

Thus the total probability of aborting the protocol is
\begin{eqnarray*}
P_{abort} & = & P_{test} + P_{EC}\\
& \leq & (1-\mu(1-\exp(2\eta^2)))^M + \eta.
\end{eqnarray*}
Hence the proposed protocol is $\epsilon^c$-complete for $\epsilon^c>(1-\mu(1-\exp(2\eta^2)))^M + \eta.$
\end{proof}

\begin{theorem}
For large $M$, Alice and Bobs can proceed for the share distribution and reconstruction phases securely if the testing phase of the proposed scheme is successful.
\end{theorem}

\begin{proof}
In the proposed scheme, there is a one to one correspondence between the entangled states $\{\ket{\psi}\}_{i=1}^M$ and the random bit string $T=\{T_1\hdots T_M\}\in\{0,1\}^M$ in such a way that, the entangled state $\ket{\psi}_i$ will be associated with the testing phase for $T_i=1$. For $i\in[M]$ such that $T_i=1$, we have defined a random variable $C_i$ by: $C_i = 1$, if they win the $N$-partite $d$-dimensional XOR game and $0$ otherwise. Now, define $C=\frac{1}{\sum_{i}T_i}(\sum_{i=0}^{M}C_i)$. Then $\mathbb{E}(C)=p_w.$ By applying Hoeffding bound~\cite{Hoeffding}, we get that
\begin{equation*}
\Pr(|C-\mathbb{E}(C)|\geq\delta)\leq\exp(-2\delta^2\sum_i T_i)=\epsilon_{test},
\end{equation*}
where $\epsilon_{test}$ is a negligibly small positive value. Thus we can express $\delta$ in term of $\epsilon_{test}$, i.e., $$\delta=\sqrt{\frac{1}{2\sum_i T_i}\ln\left(\frac{1}{\epsilon_{test}}\right)}.$$

If possible, let us define $C_i$ for $T_i=1$ as above and $\acute{C}=\frac{1}{M-\sum_{i}T_i}(\sum_{\{T_i=0\}}C_i).$ Now from the corollary of Serfling lemma~\cite{Serfling,lim} we can reduce that $\Pr(|C-\acute{C}|\geq\lambda)\leq\epsilon_{qss}$, where $\epsilon_{qss}$ is a small quantity and $$\lambda=\sqrt{\frac{M(\sum_iT_i+1)}{2(\sum_iT_i)^2(M-\sum_iT_i)}\ln\left(\frac{1}{\epsilon_{qss}}\right)}.$$As $\sum_i T_i\approx\mu M$, then $\delta$ and $\lambda$ will be very close to zero for large $M$. Thus we can conclude that if a randomly chosen subset of the entangled states pass the testing phase of the proposed scheme, then the remaining entangled states are in the desired form.
\end{proof}

\end{document}